\documentclass[a4paper,10pt]{article}
\usepackage[utf8]{inputenc}

\usepackage{amsmath,amsfonts,amssymb,amsthm,enumerate,color,graphicx}
\usepackage[capitalise]{cleveref}
\usepackage{fullpage}

\providecommand{\EE}{\mathbb{E}}
\providecommand{\Geom}{G}
\providecommand{\Exp}{E}
\providecommand{\bfz}{\mathbf{0}}
\providecommand{\Hypo}{\mathbf{H}}
\providecommand{\Gypo}{\mathbf{G}}
\providecommand{\bk}{\mathbf{k}}
\providecommand{\bh}{\mathbf{h}}
\providecommand{\Wone}{W_1}
\providecommand{\fr}[1]{\{#1\}}

\newtheorem{theorem}{Theorem}[section]
\newtheorem{lemma}[theorem]{Lemma}
\newtheorem{corollary}[theorem]{Corollary}
\theoremstyle{definition}
\newtheorem{definition}[theorem]{Definition}

\title{Asymptotic performance of the Grimmett--McDiarmid heuristic}
\author{Yuval Filmus}

\begin{document}

\maketitle

\begin{abstract}
Grimmett and McDiarmid suggested a simple heuristic for finding stable sets in random graphs. They showed that the heuristic finds a stable set of size $\sim \log_2 n$ (with high probability) on a $G(n,1/2)$ random graph. We determine the asymptotic distribution of the size of the stable set found by the algorithm.
\end{abstract}

\section{Introduction}

Grimmett and McDiarmid~\cite{GMD75} considered the problem of coloring $G(n,1/2)$ random graphs. As part of their solution, they suggested the following simple heuristic for finding a large stable set: scan the vertices in random order, adding to the stable set any vertex which is not adjacent to the vertices added so far. They showed that this heuristic constructs a stable set of size asymptotically $\log_2 n$ (with high probability), in contrast to the maximum stable set, whose size is asymptotically $2\log_2 n$ (with high probability).

Let us briefly indicate how to analyze the algorithm (for more details, consult any lecture notes on the subject). Denote by $N_k$ the number of remaining vertices not adjacent to the first $k$ vertices in the stable set constructed by the algorithm, or zero if the algorithm terminated before choosing $k$ vertices. A simple induction shows that $\EE[N_k] \leq n/2^k$, and so with high probability, the algorithm produces a stable set of size at most $\log_2 n + f(n)$, where $f(n)$ is \emph{any} function satisfying $f(n) \to \infty$.

For the lower bound, let us imagine that there are infinitely many vertices (this idea already appears in~\cite{GMD75}), let $i_0 = 0$, and let $i_k$ be the index of the $k$'th chosen vertex in the random order of the vertices (starting with~$1$). Then $i_{k+1} - i_k \sim \Geom(2^{-k})$ (geometric random variable with success probability $2^{-k}$), and the size of the clique is the maximal $k$ such that $i_k \leq n$. It is easy to calculate $\EE[i_k] = 2^k - 1$, from which it easily follows that with high probability, the algorithm produces a stable set of size at least $\log_2 n - f(n)$, where $f(n)$ is \emph{any} function satisfying $f(n) \to \infty$.

Let $\bk$ be the size of the stable set produced by the algorithm. The foregoing suggests that $\bk - \log_2 n$ approaches a limiting distribution, but there is a complication: $\bk$ is always an integer, while the fractional part of $\log_2 n$ varies. We will show that if we fix the fractional part $\fr{\log_2 n}$ then $\bk - \log_2 n$ indeed approaches a limit; and furthermore, the various limits stem from the same continuous distribution.

\begin{definition} \label{def:hypo}
 The random variable $\Hypo$ is given by the following sum of exponential distributions:
\[
 \Hypo = \sum_{i=1}^\infty \Exp(2^i).
\]
 (This defines a random variable due to Kolmogorov's two-series theorem.)
\end{definition}

\begin{theorem} \label{thm:main}
 For a given $n$, define
\[
 p_k = \Pr[\bk = k], \quad q_k = \Pr\left[\frac{n}{2^{k+1}} \leq \Hypo < \frac{n}{2^k}\right].
\]
 Then we have
\[
 \sum_{k=0}^\infty |p_k - q_k| = o(1). 
\]
\end{theorem}


\paragraph{Preliminaries} The Wasserstein distance $\Wone(X,Y)$ between two random variables is the minimum of $\EE[|X-Y|]$ over all couplings of $X,Y$. This formula shows that $\Wone(X_1+X_2,Y_1+Y_2) \leq \Wone(X_1,Y_1) + \Wone(X_2,Y_2)$. The Wasserstein distance is also given by the explicit formula
\[
 \Wone(X,Y) = \int_{-\infty}^\infty |\Pr[X < t] - \Pr[Y < t]| \, dt.
\]

The Kolmogorov--Smirnov distance between $X$ and $Y$ is $\sup_t |\Pr[X < t] - |Pr[Y < t]|$. If $Y$ is a continuous random variable with density bounded by $C$, then the Kolmogorov--Smirnov distance between $X$ and $Y$ is bounded by $2\sqrt{C\Wone(X,Y)}$.

\section{Proof} \label{sec:proof}

Recall that $\bk$ is the size of the stable set produced by the Grimmett--McDiarmid algorithm. Grimmett and McDiarmid proved the following result, whose proof was outlined in the introduction.

\begin{lemma} \label{lem:formula-bk}
 \[
  \Pr[\bk < k] = \Pr[\Geom(1) + \Geom(1/2) + \cdots + \Geom(1/2^{k-1}) > n] = \Pr[\Geom(1/2) + \cdots + \Geom(1/2^{k-1}) \geq n].
 \]
\end{lemma}

Our main idea is to rewrite this formula as follows:
\begin{equation} \label{eq:formula-bk}
 \Pr[\bk < k] = \Pr\left[\frac{\Geom(1/2^{k-1})}{n} + \frac{\Geom(1/2^{k-2})}{n} + \cdots + \frac{\Geom(1/2)}{n} \geq 1\right].
\end{equation}
It is known that the distribution $\Geom(c/n)/n$ tends (in an appropriate sense) to an exponential random variable $\Exp(c)$. We will show this quantitatively, in terms of the Wasserstein metric~$\Wone$.

\begin{lemma} \label{lem:wasserstein-g-e}
 If $p \leq 1/2$ then
 \[
  \Wone(\Geom(p)/n, \Exp(pn)) = O\left(\frac{1}{n}\right).
 \]
\end{lemma}
\begin{proof}
 Let $X = \lceil \Exp(pn)n \rceil$. Then for integer $t$,
\[
 \Pr[X \geq t] = \Pr[\Exp(pn) > (t-1)/n] = e^{-p(t-1)}.
\]
 In contrast,
\[
 \Pr[\Geom(p) \geq t] = (1-p)^{t-1}.
\]
 By construction, $\Wone(X/n, \Exp(pn)) \leq 1/n$, and so
\begin{multline*}
 \Wone(\Geom(p)/n, \Exp(pn)) \leq \frac{1}{n} + \Wone(\Geom(p)/n, X/n) \leq \frac{1}{n} + \int_0^\infty |\Pr[\Geom(p)/n \geq s] - \Pr[X/n \geq s]| \, ds = \\
 \frac{1}{n} + \frac{1}{n} \sum_{r=0}^\infty |\Pr[\Geom(p) \geq r] - \Pr[X \geq r]| =
 \frac{1}{n} + \frac{1}{n} \sum_{t=1}^\infty |(1-p)^t - e^{-pt}|.
\end{multline*}
 Since $p \leq 1/2$, we have $-p-O(p^2) \leq \log(1-p) \leq -p$, and so
\[
 e^{-pt-O(p^2t)} \leq (1-p)^t \leq e^{-pt}.
\]
 Therefore
\[
 |(1-p)^t - e^{-pt}| = e^{-pt} (1 - e^{-O(p^2t)}) = O(p^2t e^{-pt}).
\]
 We can thus bound
\[
 \sum_{t=1}^\infty |(1-p)^t - e^{-pt}| \leq O(p^2) \sum_{t=1}^\infty \frac{t}{e^{pt}} = O\left(\frac{p^2e^p}{(e^p-1)^2}\right) = O(1). \qedhere
\]
\end{proof}

%

Since $\Wone$ is subadditive, we immediately conclude the following:

\begin{lemma} \label{lem:wasserstein}
 Let $\Gypo$ be the random variable appearing in~\eqref{eq:formula-bk}. Then
\[
 \Wone\left(\frac{n}{2^k} \Gypo,\Hypo\right) = O\left(\frac{k}{2^k}\right).
\]
\end{lemma}
\begin{proof}
 \cref{lem:wasserstein-g-e} shows that
\[
 \Wone(\Gypo, \Exp(n/2^{k-1}) + \cdots + \Exp(n/2)) = O\left(\frac{k}{n}\right),
\]
 which implies that
\[
 \Wone\left(\frac{n}{2^k} \Gypo, \Exp(2) + \cdots + \Exp(2^{k-1})\right) = O\left(\frac{k}{2^k}\right).
\]
 On the other hand,
\[
 \Wone\left(\sum_{\ell=k}^\infty \Exp(2^\ell),\bfz\right) = \EE\left[\sum_{\ell=k}^\infty \Exp(2^\ell)\right] = \frac{1}{2^{k-1}},
\]
 where $\bfz$ is the constant zero random variable.
 The lemma follows. 
\end{proof}

In order to convert this bound to a bound on the Kolmogorov--Smirnov distance, we need to know that $\Hypo$ is continuous and has a bounded density function.

\begin{lemma} \label{lem:bounded-density}
 The random variable $\Hypo$ is continuous, and has a bounded density function $f$:
\[
 f(x) = 2C^{-1} \sum_{i=1}^\infty (-1)^{i-1} e^{-2^ix} \prod_{r=1}^{i-1} \frac{2}{2^r-1}, \text{ where } C = \prod_{s=1}^\infty (1-2^{-s}) > 0.
\]
 (The constant $C$ is the limit of the probability that an $n \times n$ matrix over $\mathit{GF}(2)$ is regular.)
\end{lemma}
\begin{proof}
 Let $\Hypo^{(\ell)} = \sum_{i=1}^\ell \Exp(2^i)$. It is well-known that the density of $\Hypo^{(\ell)}$ is
\[
 f_\ell(x) = \sum_{i=1}^\ell 2^i e^{-2^i x} K_{\ell,i}, \text{ where } K_{\ell,i} = \prod_{\substack{j=1\\j\neq i}}^{\ell} \frac{2^j}{2^j - 2^i}.
\]
 Note that
\[
 K_{\ell,i} = (-1)^{i-1} \prod_{j=1}^{i-1} \frac{1}{2^{i-j}-1} \times \prod_{j=i+1}^\ell \frac{1}{1-2^{i-j}} =
 (-1)^{i-1} \prod_{r=1}^{i-1} \frac{1}{2^r-1} \times \prod_{s=1}^{\ell-i} \frac{1}{1-2^{-s}}.
\]
 We can therefore write
\[
 f_\ell(x) = \sum_{i=1}^\ell 2e^{-2^i x} \times (-1)^{i-1} \prod_{r=1}^{i-1} \frac{2}{2^r-1} \times \prod_{s=1}^{\ell-i} \frac{1}{1-2^{-s}}.
\]
 This allows us to bound
\[
 |f_\ell(x)| \leq 2C^{-1}e^{-2x}\sum_{i=1}^\ell \prod_{r=1}^{i-1} \frac{2}{2^r-1},
\]
 where $C$ is the constant in the statement of the lemma.
 Bounding the sum by a geometric series, we conclude that $|f_\ell(x)| = O(e^{-2x})$, where the bound is independent of $\ell$. Applying dominated convergence, we obtain the formula in the statement of the lemma.
\end{proof}

Armed with this information, we can finally estimate $\Pr[\bk < k]$.

\begin{lemma} \label{lem:kolmogorov-smirnov}
 \[
  \Pr[\bk < k] = \Pr\left[\Hypo \geq \frac{n}{2^k}\right] \pm O\left(\sqrt{\frac{k}{2^k}}\right).
 \]
\end{lemma}
\begin{proof}
 Since $\Hypo$ has bounded density by \cref{lem:bounded-density}, we can bound the Kolmogorov--Smirnov distance between $\tfrac{n}{2^k} \Gypo$ and $\Hypo$ by $O(\sqrt{\Wone(\tfrac{n}{2^k} \Gypo, \Hypo)}) = O(\sqrt{k/2^k})$, using \cref{lem:wasserstein}. It follows that
\[
 \Pr[\bk < k] = \Pr\left[\frac{n}{2^k} \Gypo \geq \frac{n}{2^k}\right] = \Pr\left[\Hypo \geq \frac{n}{2^k}\right] \pm O\left(\sqrt{\frac{k}{2^k}}\right). \qedhere
\]
\end{proof}

\cref{thm:main} now easily follows:

\begin{proof}[Proof of \cref{thm:main}]
 \cref{lem:kolmogorov-smirnov} shows that for each $k$,
\[
 \Pr[\bk = k] = \Pr[\bk < k+1] - \Pr[\bk < k] =
 \Pr\left[\frac{n}{2^{k+1}} \leq \Hypo < \frac{n}{2^k}\right] \pm O\left(\sqrt{\frac{k}{2^k}}\right).
\]
 This implies that
\[
 \sum_{k=\ell}^\infty \left| \Pr[\bk = k] - \Pr\left[\frac{n}{2^{k+1}} \leq \Hypo < \frac{n}{2^k}\right] \right| = O\left(\sqrt{\frac{\ell}{2^\ell}}\right).
\]

 \cref{lem:formula-bk} shows that
\[
 \Pr[\bk < \ell] = \Pr[\Geom(1/2) + \cdots + \Geom(1/2^{\ell-1}) \geq n] \leq \frac{\EE[\Geom(1/2) + \cdots + \Geom(1/2^{\ell-1})]}{n} < \frac{2^\ell}{n},
\]
 and so choosing $\ell := \tfrac{2}{3} \log_2 n$, we have
\[
 \Pr[\bk < \ell] \leq \frac{1}{n^{1/3}}.
\]
 \cref{lem:kolmogorov-smirnov} shows that
\[
 \Pr\left[\Hypo \geq \frac{n}{2^\ell}\right] = O\left(\frac{\sqrt{\log n}}{n^{1/3}}\right),
\]
 and so
\[
 \sum_{k=0}^{\ell-1} \left| \Pr[\bk = k] - \Pr\left[\frac{n}{2^{k+1}} \leq \Hypo < \frac{n}{2^k}\right] \right| \leq
 \sum_{k=0}^{\ell-1} \left(\Pr[\bk = k] + \Pr\left[\frac{n}{2^{k+1}} \leq \Hypo < \frac{n}{2^k}\right] \right) = O\left(\frac{\sqrt{\log n}}{n^{1/3}}\right).
\]
 In total, we conclude that
\[
\sum_{k=0}^\infty \left| \Pr[\bk = k] - \Pr\left[\frac{n}{2^{k+1}} \leq \Hypo < \frac{n}{2^k}\right] \right| = O\left(\frac{\sqrt{\log n}}{n^{1/3}}\right). \qedhere
 \]
\end{proof}

We can also express \cref{thm:main} in terms of the variation distance between $\bk$ and an appropriate random variable.

Let $\theta = \fr{\log_2 n} = \log_2 n - \lfloor \log_2 n \rfloor$, and let $k = \lfloor \log_2 n \rfloor + c$. Then $n/2^k = 2^{\theta - c}$, and so the quantity $q_k$ in \cref{thm:main} is
\[
 \Pr[2^{-(c+1)} \leq 2^{-\theta} \Hypo < 2^{-c}] =
 \Pr[2^{-(c+1)} < 2^{-\theta} \Hypo \leq 2^{-c}] =
 \Pr[\lfloor \log_2 (1/\Hypo) + \theta \rfloor = c].
\]
Therefore we obtain the following corollary:

\begin{corollary} \label{cor:main}
 For a given $n$, let $\theta = \fr{\log_2 n}$ and define
\[
 \bh = \lfloor \log_2 (1/\Hypo) + \theta \rfloor.
\]
 The variation distance between $\bk$ and $\bh$ is at most $\tilde O(1/n^{1/3})$.
\end{corollary}

The random variable $\log_2 (1/\Hypo)$ has density
\[
 g(y) = (2 C^{-1} \ln 2) 2^{-y} \sum_{i=1}^\infty (-1)^{i-1} e^{-2^{i-y}} \prod_{r=1}^{i-1} \frac{2}{2^r-1},
\]
and is plotted in \cref{fig:density}.

\begin{figure}
 \centering
 \includegraphics[scale=0.5]{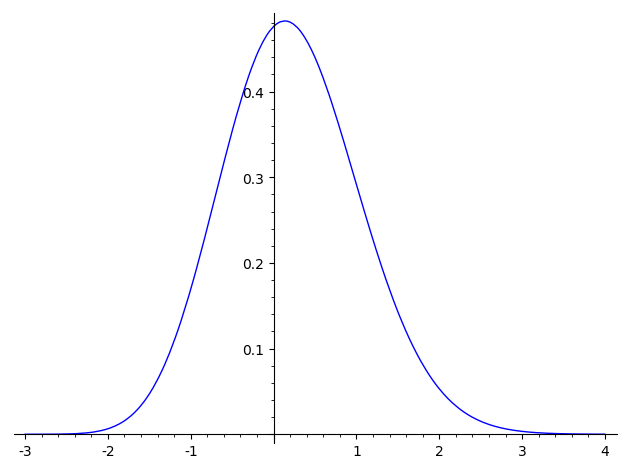}

 \caption{Density of $\log_2 (1/\Hypo)$}
 \label{fig:density}
\end{figure}

\section{Applications} \label{sec:formulas}

Integrating the formula given in \cref{lem:bounded-density}, we obtain the following estimate via \cref{lem:kolmogorov-smirnov}:
\[
 Pr[\bk = k] \approx C^{-1} \sum_{i=1}^\infty (-1)^{i-1} \left(e^{-n2^{i-k-1}} - e^{-n2^{i-k}} \right) \prod_{r=1}^{i-1} \frac{1}{2^r-1},
\]
where the error is $O(k/2^k)$. If $k = \log_2 n + c$, then this becomes
\[
 Pr[\bk = \log_2 n + c] \approx C^{-1} \sum_{i=1}^\infty (-1)^{i-1} \left(e^{-2^{i-c-1}} - e^{-2^{i-c}} \right) \prod_{r=1}^{i-1} \frac{1}{2^r-1}.
\]

Using this, we can calculate the limiting distribution of $\bk$, fixing $\fr{\log_2 n}$. For example, if $n$ is a power of~$2$ then we obtain the following limiting distribution:
\[
 \begin{array}{r|l}
  c & \lim \Pr[\bk = \log_2 n + c] \\\hline
  -4 & 0.000000389680708123307 \\
  -3 & 0.00116084271918975 \\
  -2 & 0.0610996920580558 \\
  -1 & 0.343335642221465 \\
  0 & 0.420730421531672 \\
  1 & 0.153255882765631 \\
  2 & 0.0194547690538043 \\
  3 & 0.000943671851018291 \\
  4 & 0.0000185343323798604 \\
  5 & 0.000000153237063593714
 \end{array}
\]

In this case, the expected deviation of $\bk$ from $\log_2 n$ is $-0.273947769982407$, and the standard deviation of $\bk$ is $0.763009254799132$.

\bibliographystyle{alpha}
\bibliography{biblio}

\end{document}